\documentclass[%
preprintnumbers,
nofootinbib,
 amsmath,amssymb,
 aps,
]{revtex4-2}

\usepackage{graphicx}% Include figure files
\usepackage{dcolumn}% Align table columns on decimal point
\usepackage{bm}% bold math
\usepackage{array}
\usepackage{booktabs}
\usepackage{tabu}
\usepackage{footnote}
\usepackage{dcolumn}
\usepackage{amsmath}
\usepackage{amsfonts}
\usepackage{amssymb}
\usepackage{graphicx}
\usepackage{subfigure}
\usepackage{graphicx}% Include figure files
\usepackage{dcolumn}% Align table columns on decimal point
\usepackage{bm}% bold math
\usepackage{xcolor}

\usepackage{amsthm}

\newtheorem{theorem}{Theorem}
\newtheorem{corollary}{Corollary}[theorem]

\begin{document}

%\preprint{APS/123-QED}

\title{The tripartite quantum-memory-assisted entropic uncertainty relation and upper bound on shareability of quantum discord}% Force line breaks with \\
%\thanks{A footnote to the article title}%

\author{Hazhir Dolatkhah}
  \email{h.dolatkhah@gmail.com}
  \affiliation{RCQI, Institute of physics, Slovak Academy of Sciences, \\D\'{u}bravsk\'{a} Cesta 9, 84511 Bratislava, Slovakia}
  \affiliation{Department of Physics, University of Kurdistan, P.O.Box  66177-15175, Sanandaj, Iran}
 %\altaffiliation[Also at ]{Department of Physics, University of Kurdistan, P.O.Box  66177-15175, Sanandaj, Iran}%Lines break automatically or can be forced with \\
%\author{Second Author}%
% \email{Second.Author@institution.edu}
%\affiliation{%
% Authors' institution and/or address\\
% This line break forced with \textbackslash\textbackslash
%}%

%\collaboration{MUSO Collaboration}%\noaffiliation
\author{Abolhassan Mohammadi}
 %\homepage{http://www.Second.institution.edu/~Charlie.Author}
\affiliation{Department of Physics, University of Kurdistan, P.O.Box  66177-15175, Sanandaj, Iran}%
\author{Soroush Haseli}
\affiliation{Faculty of Physics, Urmia University of Technology, Urmia, Iran.}

%\collaboration{CLEO Collaboration}%\noaffiliation

\date{\today}% It is always \today, today,
             %  but any date may be explicitly specified

\begin{abstract}
Quantum discord and quantum uncertainty are two important features of the quantum world. In this work, the relation between entropic uncertainty relation and the shareability of quantum discord is studied. First, by using tripartite quantum-memory-assisted entropic uncertainty relation, an upper bound for the shareability of quantum discord among different partites of a composite system is obtained. It is also shown that, for a specific class of tripartite states, the obtained relation could be expressed as monogamy of quantum discord. Moreover, it is illustrated that the relation could be generalized and an upper bound for the shareability of quantum discord for multipartite states is derived.
\end{abstract}

%Use showkeys class option if keyword
                              %display desired
\maketitle

%\tableofcontents

\section{\label{sec:intro}Introduction}

The uncertainty principle plays a crucial role in the field of quantum mechanics and it is known to be one of the fundamental concepts of the quantum world \cite{Heisenberg}. In quantum information theory, the uncertainty principle could be expressed in terms of the Shannon entropy. The entropy was used by Deutsch, as a criterion of uncertainty, which led to the formulation of the most famous form of the entropic uncertainty relation (EUR) \cite{Deutsch}. The Deutsch's uncertainty bound was modified by Kraus \cite{Kraus}, and a year later, it was proved by Maassen and Uffink \cite{Uffink}. The relation states that for two incompatible observables $X$ and $Z$, the following EUR will hold
\begin{equation}\label{eur}
    H(X) + H(Z) \geq \log_2{1 \over c} \equiv q_{MU} ,
\end{equation}
in which $H(Q) = - \Sigma_k p_k \log_2p_k$ is the Shannon entropy of the measurable $Q \in \{X,Z\}$, $p_k$ stands for the probability of the outcome $k$, and the parameter $c$ is defined as $c = max_{\{\mathbb{X},\mathbb{Z}\}} |\langle x_i| z_j \rangle|^2$, where $\mathbb{X}=\{ |x_i\rangle \}$ and $\mathbb{Z}=\{ |z_j\rangle \}$ are the eigenstates of the observables $X$ and $Z$, respectively. Also, $q_{MU}$ is addressed as the incompatibility measure.  \\

Expanding and modifying the relation is one of the main purposes in the field of quantum information which is being pursued by many researchers \cite{Berta,Renes,Coles1,Bialynicki,Wehner,Pati,Ballester,Vi,Wu,Rudnicki,Pramanik,Maccone,Coles,Adabi,Adabi1,Dolat2,Dolatkhah,Haseli2,Zozor,R,Kamil,Rudnicki1,Pramanik1}.
In \cite{Berta}, it was found that using the memory particle, the entropic uncertainty could be decreased. It resulted in a new uncertainty relation known as bipartite quantum-memory-assisted entropic uncertainty relation (QMA-EUR). The relation is read as
\begin{equation}\label{qmaeur}
    S(X|B) + S(Z|B) \geq q_{MU} + S(A|B),
\end{equation}
in which $S(A|B)$ is the conditional von-Neumann entropy of $\rho_{AB}$, and $S(\mathcal{O}|B) = S(\rho_{\mathcal{O}B}) - S(\rho_B), \mathcal{O} \in \{X,Z\}$ are the conditional von-Neumann entropies of the post-measurement states after measuring $X$ and $Z$ on the part $A$,
\begin{eqnarray}
\rho_{XB} & = & \sum_i \left( |x_i \rangle \langle x_i |_A \otimes I_B \right) \rho_{AB}
               \left( |x_i \rangle \langle x_i |_A \otimes I_B \right), \\
\rho_{ZB} & = & \sum_j \left( |z_j \rangle \langle z_j |_A \otimes I_B \right) \rho_{AB}
               \left( |z_j \rangle \langle z_j |_A \otimes I_B \right).
\end{eqnarray}

The bipartite QMA-EUR could be extended to the tripartite QMA-EUR \cite{Renes,Berta}, where the quantum memories are played by two extra particles $B$ and $C$. In tripartite QMA-EUR, a quantum state $\rho_{ABC}$ is shared by Alice, Bob, and Charlie, so that Alice, Bob and Charlie have access to parts $A$, $B$, and $C$, respectively. Then, Alice carries the measurement $X$ or $Z$ on her quantum system. Suppose that Alice measures $X$. Then, it is Bob's job to minimize his uncertainty about $X$. On the other hand, if Alice measures $Z$, then it would be Charlie's task to minimize his uncertainty about $Z$. The tripartite QMA-EUR is given by \cite{Renes,Berta},
\begin{equation}\label{triQMA}
S(X|B) + S(Z|C) \geq q_{MU}.
\end{equation}

Some effort has been put into modifying and improving the bound presented in Eq.\eqref{triQMA} \cite{Ming,Dolat}. In \cite{Dolat}, the lower bound of the tripartite QMA-EUR is improved by adding two additional terms to the lower bound of the relation as
\begin{equation}\label{Hajir1}
S(X|B) + S(Z|C) \geq q_{MU} + {S(A|B) + S(A|C) \over 2} + {\rm max}\{0,\delta \},
\end{equation}
where 
$$\delta = {1 \over 2} [I(A:B)+I(A:C)] - [I(X:B)+I(Z:C)],$$ in which 
$I(A:B)$ and $I(P:B)$ respectively are mutual information and Holevo quantity and are given by
\begin{eqnarray}
I(A:B) & = &  S(\rho_A) + S(\rho_B) - S(\rho_{AB}),  \\
I(P:B) & = &  S(\rho_B) - \sum_i p_i S(\rho_{B|i}).
\end{eqnarray}
Note that, as the observable $P$ on the part $A$ is measured by Alice, the $i$-th outcome is obtained with probability $p_{i} = Tr_{AB} (\Pi_{i}^{A}\rho_{AB} \Pi_{i}^{A})$ and the part $B$ is left in the corresponding state
$S(\rho_{B|i}) =\frac{Tr_{A} (\Pi_{i}^{A}\rho_{AB} \Pi_{i}^{A})}{p_{i}}$. Recently, it is shown that this lower bound, Eq.\;\eqref{Hajir1}, is tighter than the bounds that have been introduced so far \cite{Dolat,sr2021,hadwang}.\\

The topics we discussed so far were EURs and QMA-EURs with only two observables. However, QMA-EURs could be generalized to more than two observables. It has been the subject of many researches and up to now, many QMA-EURs for more than two observables have been introduced \cite{Dolatkhah,Liu,Zhang,Yunlong,multi2,multi3,multi4,multi5,multi6,multi7}. For instance, new QMA-EUR for multipartite systems has been proposed in \cite{hadsci}, where the memory is divided into multiple parts, as follows

 \begin{equation}\label{Haddadi}
    \sum_{m=1}^{N} S(M_m|X_m) \geq - \log_{2}(b) + {N-1 \over N} \sum_{m=1}^{N} S(A|X_m)
    +{\rm max}\{0,\delta^{N} \} ,
\end{equation}
in which
\begin{equation}
b=\max_{i_{N}} \left\{ \sum_{i_{2}\sim{i_{N-1}}} \max_{i_{1}} \Big[ |\langle u^{1}_{i_{1}}|u^{2}_{i_{2}}\rangle |^{2} \Big]\prod^{N-1}_{m=2}|\langle u^{m}_{i_{m}}|u^{m+1}_{i_{m+1}}\rangle |^{2} \right\} ,
\end{equation}
where $|u^{m}_{i_{m}}\rangle$ is the $i$-th eigenvector of $M_{m}$, and  $\delta^{N}={N-1 \over N} \sum_{m=1}^{N} I(A:X_m) - \sum_{m=1}^{N} I(M_m:X_m).$ $M_m$ indicates the different incompatible observables and $X_m$ stands for the memory particles for $m$-th measurement. In this uncertainty game, a multipartite quantum state $\rho_{AX_{1}...X_{N}}$ is shared by Alice and the others. Now, Alice measures one of the observables $M_m (m=1,2,...,N)$ on her quantum system. As Alice measures the observable $M_m$, the $X_m$'s task will be minimization his uncertainty about $M_{m}$. \\

 The QMA-EUR has been realized to have potential applications in various quantum information processing tasks, such as   quantum key distribution \cite{Berta,Koashi}, quantum metrology \cite{Giovannetti}, quantum cryptography \cite{Dupuis,Koenig}, quantum randomness \cite{Vallone,Cao}, entanglement witness \cite{Berta2,Huang}, EPR steering \cite{Walborn,Schneeloch},  and so on.\\
Moreover, many authors attempted to establish the relationship between quantum correlations and entropic uncertainty relations \cite{r39,r40,r41,r46,r47,r48,r50,r53,r60,r65,r66,r69,r70,n2,n4,rssm,rssm2,rssm3,new02,new04,
new05,new06,new07,new01,wangr}. Recently, Hu and Fan could obtained a new upper bound for quantum discord (QD) using bipartite QMA-EUR \cite{r39}. Also, they extracted an upper bound for the shareability of QD.    \\
In this paper, inspiring from \cite{r39} and by using tripartite QMA-EUR, an upper bound shareability of QD will be found. In the beginning, new relations for tripartite QMA-EUR are introduced. Then, it is shown that by using these relations, one could obtain a new upper bound for the shareability of QD. Also, it is shown that for specific states, the obtained relation could be considered as monogamy of QD. Finally, it is exhibited that the above procedure could be generalized to a multipartite system, in which an upper bound for the shareability of QD in a multipartite system is derived. \\

The paper has been organized as follows: In Sec.II, the QD will be defined as one of the measures of quantum correlation. In Sec.III, the new relation for the tripartite QMA-EUR is expressed and also an upper bound for the shareability of QD is extracted. The results will be summarized in Sec.IV.

%%%%%%%%%%%%%%%%%%%%%%%%%%%%%%%%%%%%%%%%%%%%%%%%%%%%%
%%%%%%%%%%%%%%%%%%%%%%%%%%%%%%%%%%%%%%%%%%%%%%%%%%%%%

%%%%%%%%%%%%%%%%%%%%%%%%%%%%%%%%%%%%%%%%%%%%%%%%%%%%%
%%%%%%%%%%%%%%%%%%%%%%%%%%%%%%%%%%%%%%%%%%%%%%%%%%%%%
\section{Quantum Discrod}
Another important concept in the field of quantum information is QD. QD and also its possible connections with other aspects of quantum information and beyond, including quantum communication, quantum computation, many-body physics, and open quantum dynamics have received huge attention (refer to \cite{Bera_QD} for more detail). \\
%so that it has been explored in many different aspects such as local creativity \cite{Streltsov_QD,Hu_QD,Gessner_QD,Abad_QD}, operational interpretation \cite{Madhok,Cavalcanti} and so on
The concept of QD of a bipartite quantum system is defined in several ways which could be classified into two wide categories. One of these categories is based on measurement in any one of the subsystems which will be used in our discussion. \\ 
QD is the difference between the total and the classical correlations \cite{Ollivier,Henderson}, namely,
\begin{equation}
D_A(\rho_{AB})=I(\rho_{AB})-J_A(\rho_{AB}),
\end{equation}
in which the subscript of $D_A(\rho_{AB})$ denotes that the measurement has been performed on the subsystem $A$. The total correlations in state $\rho_{AB}$ measured by the quantum mutual information is defined by
\begin{equation}
I(\rho_{AB}) = S(\rho_A) + S(\rho_B) - S(\rho_{AB}),
\end{equation}
and the classical correlation $J_A(\rho_{AB})$, given by
\begin{equation}
J_A(\rho_{AB}) = S(\rho_B)-min_{\Pi_i^A} S(\rho_{B|{\Pi_i^A}}),
\end{equation}
where $S(\rho_{B|{\Pi_i^A}})=\sum_i p_i S(\rho_{B|i})$ and the minimization is taken over all quantum measurements, ${\Pi_i^A }$, performed on the system $A$. \\

Recently, Hu and Fan have investigated a relation between QD and bipartite QMA-EUR \cite{r39}. Their consideration led to an improvement on the upper bounds for QD \cite{r39}. They also considered the effects of the bipartite QMA-EUR on the shearability of quantum correlation among different subsystems. \\
With the use of the bipartite QMA-EUR, Hu and Fan found an upper bound on the shearability of QD among different parties of a composite system, which is given by \cite{r39}
\begin{equation}
D_A(\rho_{AB}) + D_A(\rho_{AC}) \leq S(\rho_A) + \delta_T ,
\end{equation}
where $\delta_T = S(Q|B) + S(R|B) - \log_2{1 \over c} - S(A|B)$. They showed that for any tripartite state $\rho_{ABC}$ with $S(\rho_A) = -S(A|BC)$, the above relation can be written as:
\begin{equation}\label{Hu11}
D_A(\rho_{AB}) + D_A(\rho_{AC}) \leq D_A(\rho_{A:BC}) + \delta_T .
\end{equation}
As mentioned in \cite{r39}, one can address Eq.\;\eqref{Hu11} as the released version of the monogamy relation of QD, applicable for all tripartite pure states and to extended classes of mixed states.

%%%%%%%%%%%%%%%%%%%%%%%%%%%%%%%%%%%%%%%%%%%%%%%%%%%%%
%%%%%%%%%%%%%%%%%%%%%%%%%%%%%%%%%%%%%%%%%%%%%%%%%%%%%
%%%%%%%%%%%%%%%%%%%%%%%%%%%%%%%%%%%%%%%%%%%%%%%%%%%%%
%%%%%%%%%%%%%%%%%%%%%%%%%%%%%%%%%%%%%%%%%%%%%%%%%%%%%
%%%%%%%%%%%%%%%%%%%%%%%%%%%%%%%%%%%%%%%%%%%%%%%%%%%%%
%%%%%%%%%%%%%%%%%%%%%%%%%%%%%%%%%%%%%%%%%%%%%%%%%%%%%
%%%%%%%%%%%%%%%%%%%%%%%%%%%%%%%%%%%%%%%%%%%%%%%%%%%%%
%%%%%%%%%%%%%%%%%%%%%%%%%%%%%%%%%%%%%%%%%%%%%%%%%%%%%
\section{tripartite QMA-EUR and shareability of QD}
In this section, inspiring by Hu and Fan \cite{r39}, who obtained an upper bound on the shareability of QD among the constituent parties by using bipartite QMA-EUR, we are going to introduce a new upper bound on the shareability of QD by utilizing tripartite QMA-EUR. For this issue, first, let us introduce new tripartite QMA-EURs.
\begin{theorem}\label{theorem1}
For any tripartite state, the following equations hold
\begin{eqnarray}
S(X|B) + S(Z|C) & \geq & q_{MU} + {1 \over 2} \; \Big[ S(A|B) + S(A|C) \Big] + {\rm max} \{ O, \delta'^{3} \}, \label{theorem1_1}\\
S(X|B) + S(Z|C) & \geq & q_{MU} + {1 \over 2} \; \Big[ S(A|B) + S(A|C) \Big] + {\rm max} \{ O, \delta''^{3} \}, \label{theorem1_2}
\end{eqnarray}
where the quantities $\delta'^{3}$ and $\delta''^{3}$ are defined as
\begin{eqnarray}
\delta'^{3} & = & {1 \over 2} \; \Big\{ D_A(\rho_{AB}) + D_A(\rho_{AC}) - J_A(\rho_{AB}) - J_A(\rho_{AC}) \Big\}, \\
\delta''^{3} & = & \; \left\{ D_A(\rho_{AB}) + D_A(\rho_{AC}) - {1 \over 2} \big[ I(A:B) + I(A:C) \big] \right\}.
\end{eqnarray}

\end{theorem}

\begin{proof}
The theorem is proved using the definition of classical correlation, QD, and tripartite QMA-EUR (Eq.\;\eqref{Hajir1}). Regarding Eq.\;\eqref{Hajir1}, one obtains
\begin{eqnarray}
S(X|B) + S(Z|C) & \geq & q_{MU} + {1 \over 2} \; \Big[ S(A|B) + S(A|C) \Big]
    + {1 \over 2} \big[ I(A:B) + I(A:C) \big] - I(X:B) - I(Z:B) \nonumber \\
   & \geq & q_{MU} + {1 \over 2} \; \Big[ S(A|B) + S(A|C) \Big]
    + {1 \over 2} \big[ I(A:B) + I(A:C) - 2J_A(\rho_{AB}) - 2J_A(\rho_{AC}) \big] \nonumber \\
                & = & q_{MU} + {1 \over 2} \; \Big[ S(A|B) + S(A|C) \Big]
    + {1 \over 2} \; \Big\{ D_A(\rho_{AB}) + D_A(\rho_{AC}) - J_A(\rho_{AB}) - J_A(\rho_{AC}) \Big\}.\nonumber \\
\end{eqnarray}
Note that in the second row of the above relation is due to the fact that $J_A(\rho_{AB}) \geq I(X|B)$ and also $J_A(\rho_{AC}) \geq I(X|C)$. In the last line of the above proof, the definition of QD has been used as well. \\
The other equation of the theorem is proved by following the same procedure.
\end{proof}

Here it should be mentioned that it is clear that Eq.\;\eqref{Hajir1} is tighter that Eqs.\;\eqref{theorem1_1} and \eqref{theorem1_2}. \\

\begin{corollary}
Regarding Eq.\;\eqref{Haddadi} and similar to the proof of Theorem.\;\ref{theorem1}, one can generalize the above equations to the multipartite system as

 \begin{equation}\label{Haddad}
    \sum_{i=1}^{N} S(M_i|X_i) \geq - \log(b) + {N-1 \over N} \sum_{i=1}^{N} S(A|X_i)
    +  {\rm max}\{0,\delta'^{N} \},
\end{equation}
where $\delta'^{N}=\sum_{i=1}^{N} D_A(\rho_{AX_i})-{1 \over N} \sum_{i=1}^{N} I(A:X_i).$
\end{corollary}

Now, having these equations and following the same approach presented in \cite{r39}, one can obtain an upper bound for the shareability of QD among different subsystems. Using the tripartite QMA-EUR Eq.\;\eqref{theorem1_2} via the the theorem below: \\

\begin{theorem}
For any tripartite state $\rho_{ABC}$, we have
\begin{equation}\label{theorem2}
    \Delta_1 + \Delta_2 + S(A) \geq D_A(\rho_{AB}) + D_A(\rho_{AC}),
\end{equation}
where
\begin{equation*}
    \Delta_1 = S(X|B) + S(Z|C) - q_{MU} - {1 \over 2} \; \left[ S(A|B) + S(A|C) \right],
\end{equation*}
and
\begin{equation*}
    \Delta_2 = {-1 \over 2} \; \left[ S(A|B) + S(A|C) \right].
\end{equation*}
\end{theorem}

\begin{proof}
From Eq.\;\eqref{theorem1_2}, one arrives at
\begin{equation}\label{proof2}
    S(X|B) + S(Z|C) - q_{MU} - {1 \over 2} \; \left[ S(A|B) + S(A|C) \right]
    + {1 \over 2} \; \left[ I(A:B) + I(A:C) \right]  \geq  D_A(\rho_{AB}) + D_A(\rho_{AC}).
\end{equation}
Utilizing the following relation
\begin{equation}
   { 1 \over 2} \; \left[ I(A:B) + I(A:C) \right] = S(A) - {1 \over 2} \;
   \left[ S(A|B) + S(A|C) \right] ,
\end{equation}
in Eq.\eqref{proof2}, one comes to
\begin{eqnarray}
    S(X|B) + S(Z|C) - q_{MU}  - {1 \over 2} \; \left[ S(A|B) + S(A|C) \right] + S(A)
    - {1 \over 2} \; \left[ S(A|B) + S(A|C) \right] \geq  D_A(\rho_{AB}) + D_A(\rho_{AC}).
\end{eqnarray}
Therefore, the theorem has been proved.  \\
\end{proof}

As can be seen, our relation, Eq\;\eqref{theorem2}, contains three terms in which $S(A)$ indicates the entropy of the subsystem $A$, $\Delta_1$ is related to the tripartite QMA-EUR, and $\Delta_2$ that is related to the strong subadditivity (SSA) inequality. Note that from the tripartite EUR there is $\Delta_1 \geq 0$, and from the SSA inequality one could arrive at $\Delta_2 \leq 0$. \\

\begin{corollary}
From the above theorem, it could be resulted that for any tripartite state $\rho_{ABC}$ with
$S(A) = - S(A|BC)$, we have
\begin{equation}\label{theofortri}
    D_A(\rho_{AB}) + D_A(\rho_{AC}) \leq D_A(\rho_{A:BC}) + \Delta_1 + \Delta_2.
\end{equation}
\end{corollary}

\begin{proof}
The outline of the proof is similar to what we have in \cite{r39}. Due to the fact that $D_A(\rho_{A:BC}) = S(\rho_A)$ and the
relation $S(\rho_A) = -S(A|BC)$, it is realized that  Eq.\;\eqref{theofortri} is valid for all
tripartite pure state. As stated in \cite{Xi}, under a specific condition, the
relation $S(\rho_A) = -S(A|BC)$ is reliable even for a mixed state $\rho_{ABC}$. The relation
is true for a mixed state if and only if for the Hilbert space $\mathcal{H}_{BC}$ we have a
factorization $\mathcal{H}_{BC} = \mathcal{H}_{(BC)^L} \otimes \mathcal{H}_{(BC)^R}$ in which
$\rho_{ABC} = |\psi \rangle _{A(BC)^L} \langle \psi | \otimes \rho_{(BC)^R}$. For this case, it is obtained that $D_A(\rho_{A:BC}) = D_A(|\psi \rangle _{A(BC)^L}) = S(\rho_A)$.  \\
\end{proof}

Similar to what we have in \cite{r39}, our obtained inequality \eqref{theofortri} is also realized to be a released version of the monogamy relation of QD \cite{Streltsov,Giorgi,Prabhu,Braga}, and can be utilized for all tripartite pure and extended  mixed tripartite states. \\

It should be implied that, since $S(A|B) + S(A|C) = 0$, the relation \eqref{theofortri} for the pure state ($|\psi \rangle_{ABC}$) leads to
\begin{equation}\label{theoforpure}
    D_A(\rho_{AB}) + D_A(\rho_{AC}) \leq D_A(\rho_{A:BC}) + S(X|B) + S(Z|C) - q_{MU} .
\end{equation}
Due to the fact that the following two relations
\begin{equation*}
    S(X|B) + S(Z|B) \geq q_{MU} + S(A|B) ,
\end{equation*}
and
\begin{equation*}
  S(X|B) + S(Z|C) \geq q_{MU} ,
\end{equation*}
are equivalent \cite{Berta}, it could be concluded that the upper bound of our relation \eqref{theoforpure} is equivalent to the bound of Eq.\;\eqref{Hu11} obtained by Hu and Fan \cite{r39}. However, we believe that our upper bound for mixed tripartite state $\rho_{ABC}$ with $S(A) = -S(A|BC)$ is tighter than the upper bound found by Hu and Fan \cite{r39}. This is due to the fact that in our result $\Delta_1 \leq \delta_T$ and $\Delta_2 \leq 0$. \\
We believe that our upper bound will be useful in many areas of quantum information. One of the consequences of our inequality \eqref{theofortri} is that, if for tripartite pure state $| \psi \rangle_{ABC}$ one finds two observables $X$ and $Z$ that saturate $S(X|B) + S(Z|C) \geq q_{MU}$, then it could be stated that we have the sufficient condition for the monogamy QD. The generalized Greenberger-Horne-Zeilinger (GHZ) state could be implied as one of the examples of the situation. \\
%%%%%%%%%%%%%%%%%%%%%%%%%%%%%%%%%%%%%%%%%%%
%%%%%%%%%%%%%%%%%%%%%%%%%%%%%%%%%%%%%%%%%%%
%%%%%%%%%%%%%%%%%%%%%%%%%%%%%%%%%%%%%%%%%%%
%%%%%%%%%%%%%%%%%%%%%%%%%%%%%%%%%%%%%%%%%%%
\subsection{Examples}
To clarify the above mentioned results, four examples are considered. In these examples, the observables that are measured on the part $A$ of quantum states are assumed to be the Pauli matrices $X = \sigma_{1}$ and $Z = \sigma_{3}$. \\

\subsubsection{Generalized GHZ state}\label{GHZ}
First, let us consider the  generalized GHZ states which have the form
\begin{equation}\label{GHZ}
\vert GGHZ\rangle = cos\beta \vert 000\rangle +sin\beta \vert 111\rangle,
\end{equation}
where $\beta\in\left[  0,2\pi\right)$. In Fig. \ref{fig1}, different upper bounds of the shareability of QD for these states are plotted versus the parameter $\beta$. As it was expected, the obtained upper bound coincides with Hu and Fan upper bound.
%%%%%%%%%%%%%%%%%%%%%%%%%%%%%%%%%%%
%%%%%%%%%%%%%%%%%%%%%%%%%%%%%%%%%%%
\begin{figure}[ht]
\centering
\includegraphics[width=8cm]{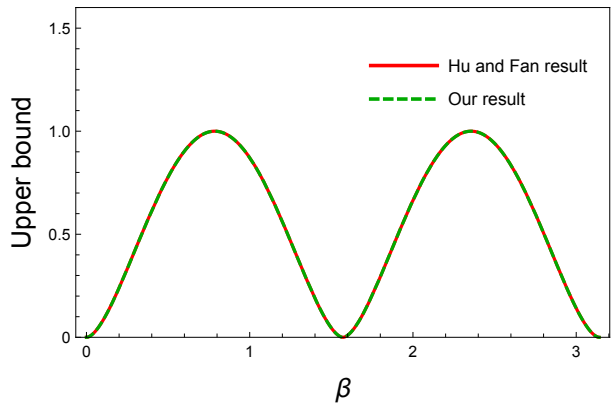}
\caption{\label{fig1} (Color online) Different upper bounds on the shareability of QD for the state in Eq.\;(\ref{GHZ}), versus the state's parameter $\beta$.}
\end{figure}
%%%%%%%%%%%%%%%%%%%%%%%%%%%%%%%%%
%%%%%%%%%%%%%%%%%%%%%%%%%%%%%%%%%%%%
\begin{figure}[ht]
\centering
\includegraphics[width=8cm]{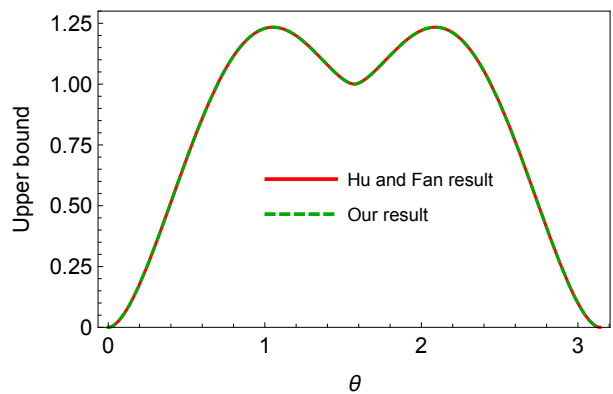}
\caption{(Color online) Different upper bounds on the shareability of QD for the state in Eq.\;(\ref{W}), versus the state's parameter $\theta$, where $\phi=\pi/4$.}\label{fig2}
\end{figure}
%%%%%%%%%%%%%%%%%%%%%%%%%%%%%%%%%

%%%%%%%%%%%%%%%%%%%%%%%%%%%%%%%%%%%%%
%%%%%%%%%%%%%%%%%%%%%%%%%%%%%%%%%%%%%
\subsubsection{Generalized W state}

As the second example, consider the following generalized $W$ state:
\begin{equation}\label{W}
\vert GW \rangle = sin\theta cos\phi \vert 100\rangle +sin\theta sin\phi \vert 010\rangle +cos\theta \vert 001\rangle ,
\end{equation}
where $\theta\in\left[  0,\pi\right]$ and $\phi\in\left[  0,2\pi\right)$. Same as the previous case, it is realized that for this state, the obtained upper bound is exactly the same as that of Hu and Fan; shown in Fig. \ref{fig2}.
%%%%%%%%%%%%%%%%%%%%%%%%%%%%%%%%%
%%%%%%%%%%%%%%%%%%%%%%%%%%%%%%%%%
%%%%%%%%%%%%%%%%%%%%%%%%%%%%%%%%%%%%%%
%%%%%%%%%%%%%%%%%%%%%%%%%%%%%%%%%%%%%%
\begin{figure}[ht]
\centering
\includegraphics[width=8cm]{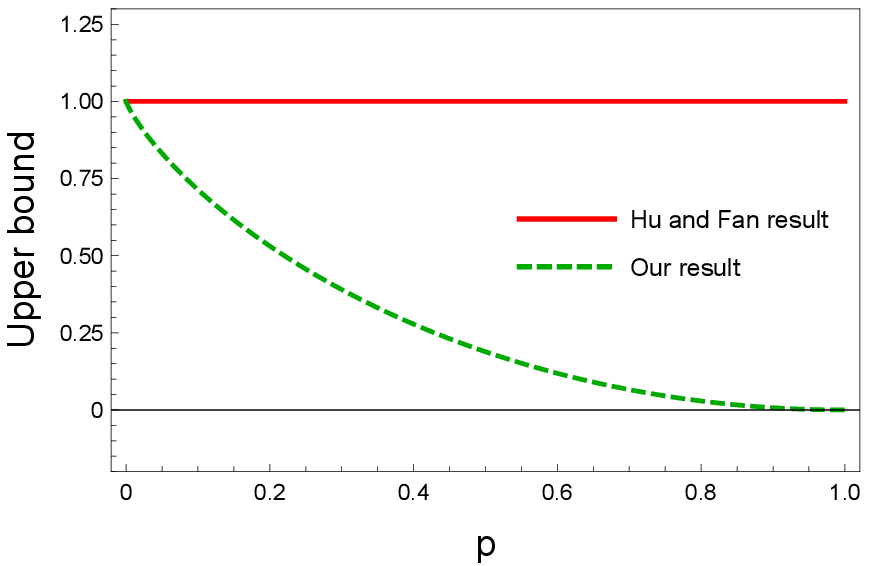}
\caption{(Color online) Different upper bounds on the shareability of QD for the state in Eq.\;(\ref{werner}), versus the state's parameter $p$, where $ 0 \leq p \leq 1$.}\label{fig3}
\end{figure}
%%%%%%%%%%%%%%%%%%%%%%%%%%%%%%%%%%%%%%%
%%%%%%%%%%%%%%%%%%%%%%%%%%%%%%%%%%%%%%%
%%%%%%%%%%%%%%%%%%%%%%%%%%%%%%%%%%%%%%
%%%%%%%%%%%%%%%%%%%%%%%%%%%%%%%%%%%%%%
\begin{figure}[ht]
\centering
\includegraphics[width=8cm]{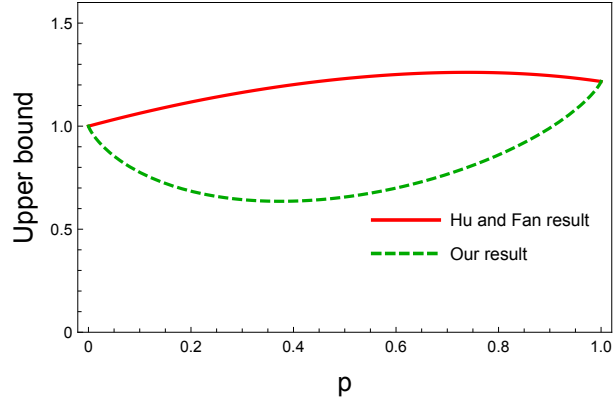}
\caption{(Color online) Different upper bounds on the shareability of QD for the state in in Eq.\;(\ref{Mix}), versus the state's parameter $p$, where $ 0 \leq p \leq 1$.}\label{fig4}
\end{figure}
%%%%%%%%%%%%%%%%%%%%%%%%%%%%%%%%%%%%%%%%%%
%%%%%%%%%%%%%%%%%%%%%%%%%%%%%%%%%%%%%%%%%%
\subsubsection{ Werner-GHZ state}
As an another example, we consider Werner-GHZ state defined as
\begin{equation}\label{werner}
\rho_{w}=(1-p) \vert GHZ \rangle \langle GHZ \vert + \frac{p}{8}\mathbf{I}_{ABC},
\end{equation}
where $\vert GHZ \rangle = 1/\sqrt{2}(\vert 000 \rangle + \vert 111 \rangle)$ is the GHZ state, and $0 \leq p \leq 1$. In Fig. \ref{fig3}, the upper bounds of the shareability of QD for this state are plotted versus the parameter $p$. As can be seen,  Hu and Fan upper bound is constant as a function of the parameter $p$, whereas our upper bound is tighter and also it decreases by enhancement of $p$.

\subsubsection{A mixed three-qubit state}
As the last example, let us consider a state of this form
\begin{equation}\label{Mix}
\rho=p\vert GHZ \rangle \langle GHZ \vert + (1-p)\vert W \rangle \langle W \vert,
\end{equation}
where $ 0 \leq p \leq 1$ is a real number and the usual $\vert W \rangle$ state is defined as
\begin{eqnarray}
\vert W \rangle &=& \frac{1}{\sqrt{3}}(\vert 001 \rangle + \vert 010 \rangle + \vert 100 \rangle). \nonumber
\end{eqnarray}
 In Fig. \ref{fig4}, the upper bounds of the shareability of QD for the state in Eq.\;(\ref{Mix}) are plotted versus the parameter $p$. It is clearly shown that, our upper bound is tighter than that of Hu and Fan.
%%%%%%%%%%%%%%%%%%%%%%%%%%%%%

%%%%%%%%%%%%%%%%%%%%%%%%%%%%%%%%%%%%%%%%%%
%%%%%%%%%%%%%%%%%%%%%%%%%%%%%%%%%%%%%%%%%%
%%%%%%%%%%%%%%%%%%%%%%%%%%%%%%%%%%%%%%%%%%
%%%%%%%%%%%%%%%%%%%%%%%%%%%%%%%%%%%%%%%%%%
\subsection{Generalization}
The obtained result could be generalized to $(N+1)$-partite states. By utilizing the multipartite uncertainty relation with quantum memory, it is possible to find an upper bound for the shareability of multipartite QD. This will be presented in the following theorem. \\

\begin{theorem}
For any $N+1$-partite state, we have
\begin{equation}\label{theorem3}
\Delta_1^N + \Delta_2^N + S(A) \geq \sum_{i=1}^{N} D_A(\rho_{AX_i}),
\end{equation}
in which
\begin{eqnarray}
\Delta_1^N & = & \sum_{i=1}^{N} S(M_i|X_i) + \log_2(b) - {N-1 \over N} \sum_{i=1}^{N} S(A|X_i), \\
\Delta_2^N & = & {-1 \over N} \sum_{i=1}^{N} S(A|X_i).
\end{eqnarray}
\end{theorem}

\begin{proof}
Regarding the Eq.\;\eqref{Haddad}, one has
\begin{equation}
    \sum_{i=1}^{N} S(M_i|X_i) + \log(b) - {N-1 \over N} \sum_{i=1}^{N} S(A|X_i)
    + {1 \over N} \sum_{i=1}^{N} I(A:X_i) \geq \sum_{i=1}^{N} D_A(\rho_{AX_i}) ,
\end{equation}
Applying the relation below
\begin{equation}
S(A) = {1 \over N} \; \left[ \sum_{i=1}^{N} S(A|X_i) + \sum_{i=1}^{N} I(A:X_i) \right] ,
\end{equation}
one comes to
\begin{equation}
\sum_{i=1}^{N} S(M_i|X_i) + \log_2(b) - {N-1 \over N} \sum_{i=1}^{N} S(A|X_i) +
 S(A) - {1 \over N} \sum_{i=1}^{N} S(A|X_i) \geq \sum_{i=1}^{N} D_A(\rho_{AX_i}) .
\end{equation}
The above equation could be rewritten as
\begin{equation}
\Delta_1^N + \Delta_2^N + S(A) \geq \sum_{i=1}^{N} D_A(\rho_{AX_i}).
\end{equation}
\end{proof}

Now, let us consider the above result for a four-partite state, i.e. $N=3$. For this case, Eq.\;\eqref{theorem3} is rewritten as
\begin{equation}
\Delta_1^3 + \Delta_2^3 + S(A) \geq D_A(\rho_{AB}) + D_A(\rho_{AC}) + D_A(\rho_{AD}),
\end{equation}
where the quantity $\Delta_1^3$ is given by
\begin{equation}
\Delta_1^3 = S(M_1|B) + S(M_2|C) + S(M_3|D) + \log_2(b') -
{2 \over 3} \; \left[ S(A|B) + S(A|C) + S(A|D) \right] ,
\end{equation}
and
\begin{equation}
b' = max_k \left\{ \sum_j max_i [|\langle u_i^1 | u_j^2 \rangle|^2] \; |\langle u_j^3 | u_k^3 \rangle|^2 \right\} ,
\end{equation}
in which $|u_i^1 \rangle$, $| u_j^2 \rangle$, and $| u_k^3 \rangle$ are the eigenstates of the three observables $M_1$, $M_2$, and $M_3$, respectively. The other quantity $\Delta_2^3$ is read as
\begin{equation}
\Delta_2^3 = {-1 \over 3} \; \Big[ S(A|B) + S(A|C) + S(A|D) \Big].
\end{equation}
Assume there is a four-partite state $\rho_{ABCD}$, where the particles A, B, C, and D are respectively sent to Alice, Bob, Charlie, and David. Alice carries only one of the three observables $M_m$ (where $m=1,2,3$) and informs the other about her choice of measurement. Then, Alice makes a measurement. Now, if Alice measures $M_1$, it is Bob's task to minimize his uncertainty about $M_1$. If $M_2$ is measured by Alice, the task is Charlie's to minimize his uncertainty about $M_2$. And for the last case, if $M_3$ is measured by Alice, it is David's task to minimize his uncertainty about $M_3$.

%%%%%%%%%%%%%%%%%%%%%%%%%%%%%%%%%%%%%%%%%%%%%%%%%%%%%
%%%%%%%%%%%%%%%%%%%%%%%%%%%%%%%%%%%%%%%%%%%%%%%%%%%%%
%%%%%%%%%%%%%%%%%%%%%%%%%%%%%%%%%%%%%%%%%%%%%%%%%%%%%
%%%%%%%%%%%%%%%%%%%%%%%%%%%%%%%%%%%%%%%%%%%%%%%%%%%%%
%%%%%%%%%%%%%%%%%%%%%%%%%%%%%%%%%%%%%%%%%%%%%%%%%%%%%
%%%%%%%%%%%%%%%%%%%%%%%%%%%%%%%%%%%%%%%%%%%%%%%%%%%%%
%%%%%%%%%%%%%%%%%%%%%%%%%%%%%%%%%%%%%%%%%%%%%%%%%%%%%
%%%%%%%%%%%%%%%%%%%%%%%%%%%%%%%%%%%%%%%%%%%%%%%%%%%%%
\section{Conclusion}
The tripartite QMA-EUR has many applications in quantum information theory; quantum key distribution could be addressed as one of these applications. Here, we presented another application of tripartite QMA-EUR. It was determined that using tripartite QMA-EUR, one could obtain an upper bound for the shareability of QD. Our bound includes three terms in which one is related to the entropy of the subsystem that is being measured. The second term is related to the tripartite QMA-EUR. And the third term implies the SSA inequality. In another word, our bound relates tripartite QMA-EUR, SSA inequality, and QD which are known as three important features of quantum information. Also, it was shown that for some tripartite states, this relation is converted to the monogamy of discord. The obtained bound could be applicable in the field of quantum information. The result indicates that for a tripartite pure state if one can find two observables $X$ and $Z$ that saturate $S(X|B) + S(Z|C) \geq q_{MU}$, then, a sufficient condition for the monogamy of QD is provided.\\
Moreover, the work could be generalized so that one could obtain an upper bound for the shareability of QD for multipartite states.

% The \nocite command causes all entries in a bibliography to be printed out
% whether or not they are actually referenced in the text. This is appropriate
% for the sample file to show the different styles of references, but authors
% most likely will not want to use it.

\end{document}